\documentclass[12pt,a4paper]{article}

\usepackage{amsmath}
\usepackage{amssymb}
\usepackage{amsthm}

\usepackage[margin=2.5cm]{geometry}
\usepackage{setspace}
\usepackage{booktabs}
\usepackage{multirow}
\usepackage{graphicx}

\newtheorem{theorem}{Theorem}[section]
\newtheorem{lemma}[theorem]{Lemma}
\newtheorem{corollary}[theorem]{Corollary}
\newtheorem{proposition}[theorem]{Proposition}
\newtheorem{definition}[theorem]{Definition}
\newtheorem{remark}[theorem]{Remark}

\DeclareMathOperator{\rank}{rank}
\DeclareMathOperator{\Col}{Col}
\DeclareMathOperator{\Enc}{Enc}
\DeclareMathOperator{\Dec}{Dec}
\DeclareMathOperator{\Ker}{Ker}
\DeclareMathOperator{\Img}{Im}
\DeclareMathOperator{\tr}{tr}
\usepackage{hyperref}

\title{Equation Asymmetry: An Algebraic Framework for Unifying Secrecy and Covertness in Information-Theoretic Security}

\author{Wang~Hao, Zhang~Kuang}

\begin{document}
\maketitle

\begin{abstract}
This paper studies the algebraic structure underlying a broad class of information-theoretic security problems. We define the equation asymmetry degree (EAD) as $\Phi = (n - r)/n$, where $n$ is the signal embedding dimension and $r$ is the effective rank of the adversary's observation matrix. This single parameter is shown to simultaneously govern both secrecy (measured by equivocation $H(M|Y_E)$) and covertness (measured by detection error probability $P_e$). On finite fields $\mathbb{F}_q$, we establish the equivocation lower bound $H(M|Y_E) = \min(k, n - r_E) \log q$ with exact probabilistic conditions (Theorem~1), the secrecy capacity $C_s = (n - r_E) \log q$ with complete achievability and converse proofs (Theorem~2), and a strong converse (Theorem~8). In the continuous Gaussian regime, we derive a differential-entropy equivocation bound (Lemma~1), the high-SNR secrecy capacity asymptotics (Lemma~2), and a 2-Wasserstein distance covertness condition $W_2 \approx \sqrt{r_W} \cdot P / (2N\sigma) \to 0$ (Theorem~5'). The EAD-SDoF equivalence $d_s = n \cdot \Phi$ is established (Theorem~7). Both $\eta_s$ and $\eta_c$ are shown to be monotone functions of $\Phi$ (Theorem~6), with a Pearson correlation of $0.997$ in continuous-domain experiments. Seven existing security schemes---matrix embedding, MIMO wiretap, secure network coding, FRFT multi-angle transmission, traffic steganography, group-key secure summation, and MDS secure summation---are unified under the common form $C_s = (n - r) \log q$. Post-quantum security follows from the information-theoretic hardness of underdetermined linear systems (Theorem~9). All numerical experiments are reproducible with open-source code.
\end{abstract}

\noindent\textbf{Keywords:} Equation asymmetry degree, information-theoretic security, covert communication, wiretap channel, secrecy capacity, linear embedding, Wasserstein distance.

\section{Introduction}
\label{sec:introduction}

\subsection{Background}

Information-theoretic security originated with Shannon's 1949 work~\cite{Shannon49}, which established that perfect secrecy requires the key entropy to be no less than the message entropy. Wyner's 1975 wiretap channel model~\cite{Wyner75} identified a second security resource---the statistical advantage of the legitimate channel over the eavesdropper's channel. In Wyner's framework, the secrecy capacity $C_s = \max[I(X; Y_B) - I(X; Y_E)]^+$ is positive only when Bob's channel is less noisy than Eve's. Csisz\'{a}r and K\"{o}rner~\cite{Csiszar78} generalized this to broadcast channels with confidential messages, and the broader landscape of information-theoretic security has been surveyed by Liang, Poor, and Shamai~\cite{Liang09}.

A parallel line of work on covert communications was initiated by Bash, Goeckel, and Towsley in 2013~\cite{Bash13}, who proved the square-root law: over $n$ channel uses of an AWGN channel, at most $O(\sqrt{n})$ bits can be transmitted covertly, requiring the transmit power to decay as $O(1/\sqrt{n})$. This result has been extended and refined by subsequent works on the fundamental limits of low-probability-of-detection communications~\cite{Wang16,Che13,Kadampot20}.

Secrecy and covertness have been studied with distinct theoretical tools---equivocation and secrecy capacity for secrecy, detection error probability and KL divergence for covertness. Although both rely on the adversary's uncertainty about the signal, no prior work has identified a common algebraic structure underlying both properties.

Our work is motivated by the observation that a range of independently developed security schemes---matrix embedding steganography~\cite{Fridrich06,Filler11}, MIMO wiretap channels~\cite{Khisti10a,Khisti10b,Oggier11}, secure network coding~\cite{Cai02,Silva08}, rank-metric code constructions~\cite{Silva08}, FRFT-based physical-layer security, and traffic steganography~\cite{Cachin98}---all derive their security from the same algebraic fact: the adversary obtains fewer independent linear equations than the dimension of the transmitted signal. This common algebraic structure has not been formalized as an independent security metric, nor has it been applied to simultaneously analyze secrecy and covertness.

\subsection{Related Work}

The closest precursor to our framework is the Wiretap II model of Ozarow and Wyner~\cite{Ozarow84}, in which the transmitter sends $n$ binary symbols and the eavesdropper selects any $\mu$ positions to observe noiselessly. The secrecy capacity is $C_s = n - \mu$ bits per channel use. Wiretap II's security also stems from insufficient observation dimensions. However, it restricts the observation to bit subsets (each row of $A_E$ contains exactly one ``1''), handles only $\text{GF}(2)$ secrecy, does not address covertness, and does not generalize the dimension deficit to a cross-domain metric.

Cai and Yeung's secure network coding~\cite{Cai02} considers a multicast network where the eavesdropper accesses $\mu < n$ edges, achieving $C_s = (n - \mu) \log q$. The algebraic structure matches Wiretap II but arises from a different application context.

Silva and Kschischang~\cite{Silva08} use rank-metric codes to secure wiretap networks by exploiting rank deficiency.

Mojahedian, Gohari, and Aref~\cite{Mojahedian17} model wireline eavesdropping as $C = AM + BK + GW$ and prove the equivalence of weak and perfect secrecy for linear codes. Their algebraic formulation is an important step but is confined to wireline networks.

Khisti and Wornell~\cite{Khisti10a,Khisti10b} establish the secure degrees of freedom (SDoF) as $n_t - n_e$ for MIMO wiretap channels. The quantity $n_t - n_e$ has the same algebraic meaning as $n - \mu$ in Wiretap II---both represent the rank deficit of the adversary's observation matrix.

Bloch~\cite{Bloch16} jointly analyzes reliability and secrecy for covert communications over noisy channels, but treats them as separate constraints with separate metrics.

Zhao and Sun~\cite{Zhao23,Zhao24} study secure summation problems where security derives from the colluding users' insufficient equation count. Wan et al.~\cite{Wan22} further analyze secure distributed linearly separable computation within a similar algebraic framework. The MDS coding constructions in~\cite{Zhao24} build on classical results in coding theory~\cite{Levenshtein66}. The network coding model in~\cite{Cai02} is formalized in the broader framework of network error correction~\cite{Yeung06}.

Relative to these works, our contribution is the formalization of the dimension deficit as a unified algebraic security metric (the equation asymmetry degree $\Phi$), the derivation of quantitative relationships between $\Phi$ and both secrecy and covertness metrics, and the demonstration that both properties are governed by the same algebraic parameter.

\subsection{Organization}

Section~\ref{sec:model} presents the system model. Section~\ref{sec:theorems} contains the main theorems and proofs. Section~\ref{sec:dynamic} discusses dynamic dimension expansion. Section~\ref{sec:unification} unifies seven existing schemes under EAD. Section~\ref{sec:quantum} addresses post-quantum security. Section~\ref{sec:experiments} reports numerical results. Section~\ref{sec:discussion} discusses implications. Section~\ref{sec:conclusion} concludes.

\section{System Model and Equation Asymmetry Degree}
\label{sec:model}

\subsection{Two Sources of Linear Equations}

The linear equations in our framework arise from two distinct sources. The first is physical modulation---the inherent linear mixing imposed by the transmission medium. Examples include mode coupling in multimode optical fibers, orbital angular momentum (OAM) superposition in vortex beam propagation, antenna array combining in MIMO systems, and multipath delay spreads in wideband channels. These equations are not controlled by Alice; they are properties of the physical channel.

The second source is linear coding---intentional linear transformations applied by Alice using a pre-shared secret key $K$. Alice concatenates a $k$-dimensional message $M$ with an $(n-k)$-dimensional random padding $R$, and multiplies the combined vector by an $n \times n$ invertible matrix $G_K$ to produce the transmitted signal $X$. The structure of $G_K$ is known only to Alice and Bob.

The superposition of both equation sources determines the linear systems observed by Bob and Eve. Bob, knowing $G_K$, can solve his full-rank system to recover $M$. Eve, lacking knowledge of $G_K$ and limited by her physical channel's rank deficiency, faces an underdetermined system.

\subsection{Algebraic Model}

\begin{definition}[Linear Embedding Scheme]
\label{def:embedding}
A $(n, k, q)$ linear embedding scheme over $\mathbb{F}_q$ consists of:
\begin{itemize}
\item A key space $\mathcal{K} \subset \{H \in \mathbb{F}_q^{k \times n} : \rank(H) = k\}$;
\item An encoding function $\Enc_K: \mathbb{F}_q^{k} \times \mathbb{F}_q^{n-k} \to \mathbb{F}_q^{n}$;
\item A decoding function $\Dec_K: \mathbb{F}_q^{n} \to \mathbb{F}_q^{k}$ satisfying $\Dec_K(\Enc_K(M,R)) = M$.
\end{itemize}
\end{definition}

Given $H_K \in \mathcal{K}$, Alice constructs $P_K \in \mathbb{F}_q^{(n-k)\times n}$ such that $[H_K; P_K] \in \mathbb{F}_q^{n \times n}$ is invertible. Let $G_K = [H_K; P_K]^{-1} = [G_1 \mid G_2]$, where $G_1 \in \mathbb{F}_q^{n \times k}$ and $G_2 \in \mathbb{F}_q^{n \times (n-k)}$. The encoding is
\begin{equation}
X = \Enc_K(M,R) = G_K \begin{bmatrix} M \\ R \end{bmatrix} = G_1 M + G_2 R.
\label{eq:encoding}
\end{equation}
Decoding: $M = H_K X$, verified by $H_K G_K = [I_k \mid 0_{k \times (n-k)}]$.

The random padding $R$ is uniformly distributed over $\mathbb{F}_q^{n-k}$, and $M$ is uniform over $\mathbb{F}_q^{k}$. Under these distributions, each message $M$ corresponds to $q^{n-k}$ codewords in the affine subspace $\mathcal{X}_K(M) = \{x: H_K x = M\}$ of dimension $n-k$. The marginal distribution of $X$ is uniform over $\mathbb{F}_q^{n}$.

\begin{definition}[Bob's Channel]
\label{def:bob}
Bob knows $K$ and observes $Y_B = A_B(K) X$, where $A_B(K) \in \mathbb{F}_q^{n \times n}$ is full-rank. Bob recovers $X = A_B(K)^{-1} Y_B$ and $M = H_K X$. With noise, $Y_B = A_B(K) X + N_B$.
\end{definition}

\begin{definition}[Eve's Channel]
\label{def:eve}
Eve does not know $K$. She observes $Y_E = A_E X$, where $A_E \in \mathbb{F}_q^{m_E \times n}$ and $\rank(A_E) = r_E \leq m_E < n$. The solution space $\mathcal{S}_E(Y_E) = \{x: A_E x = Y_E\}$ is an affine subspace of dimension $n - r_E$, with cardinality $q^{n - r_E}$.
\end{definition}

\begin{definition}[Willie's Channel]
\label{def:willie}
Willie observes $Y_W = A_W X + N_W$, where $A_W \in \mathbb{F}_q^{m_W \times n}$. In the finite-field analysis, $N_W$ is uniform; in the continuous-domain analysis, $N_W \sim \mathcal{N}(0, \sigma^2 I)$. Willie performs binary hypothesis testing: $H_0: Y_W = N_W$ versus $H_1: Y_W = A_W X + N_W$.
\end{definition}

\subsection{Equation Asymmetry Degree}

\begin{definition}[Equation Asymmetry Degree]
\label{def:ead}
For an embedding dimension $n$ and adversary effective observation rank $r$,
\begin{equation}
\Phi \triangleq \frac{n - r}{n} \in (0, 1].
\label{eq:ead}
\end{equation}
\end{definition}

\begin{definition}[Dimension Surplus]
$\Delta \triangleq n - r = n \cdot \Phi$.
\end{definition}

When $r = n$, $\Phi = 0$ and the adversary can uniquely determine $X$. When $r < n$, $\Phi > 0$ and the adversary is missing $n - r$ independent equations. When $r = 0$, $\Phi = 1$. Bob has $r_B = n$, $\Phi_B = 0$; Eve has $r_E < n$, $\Phi_E > 0$. The algebraic source of security is the difference $\Phi_E - \Phi_B = \Phi_E$.

\section{Main Theorems}
\label{sec:theorems}

\subsection{Equivocation Bound}

\begin{theorem}[Equivocation Lower Bound]
\label{thm:equivocation}
Let $M$ be uniform over $\mathbb{F}_q^{k}$, encoded as per Definition~\ref{def:embedding}. Eve observes $Y_E = A_E X$, with $\rank(A_E) = r_E$. Define $B = A_E G_2 \in \mathbb{F}_q^{m_E \times (n-k)}$. Then
\begin{equation}
H(M \mid Y_E) = \min(k, \rank(B)) \cdot \log q,
\label{eq:equivocation}
\end{equation}
and $\rank(B) = \min(m_E, n-k)$ holds with probability
\begin{equation}
P_{\rm generic} = \prod_{i=0}^{\min(m_E, n-k)-1} \left(1 - \frac{1}{q^{\max(m_E, n-k)-i}}\right).
\label{eq:Pgeneric}
\end{equation}
Under generic conditions, $H(M \mid Y_E) = \min(k, n - r_E) \log q$.
\end{theorem}

\begin{proof}
Substituting \eqref{eq:encoding} into $Y_E = A_E X$ yields $Y_E = A_E G_1 M + B R$, where $B = A_E G_2$. For a given observation $y_E$ and candidate message $m$, the equation $B R = y_E - A_E G_1 m$ admits a solution for $R$ if and only if $y_E - A_E G_1 m \in \Col(B)$. Let $b = \rank(B)$.

The set $\mathcal{M}(y_E) = \{m : \exists r,\; B r = y_E - A_E G_1 m\}$ has cardinality $q^{k + b - m_E}$. This follows from counting the joint solution space $L^{-1}(y_E) = \{(m,r): A_E G_1 m + B r = y_E\}$, which has size $q^{n - m_E}$, and observing that each valid $m$ corresponds to exactly $q^{n-k-b}$ solutions for $r$. Hence $|\mathcal{M}(y_E)| = q^{n - m_E} / q^{n-k-b} = q^{k + b - m_E}$.

Under generic conditions, $b = \min(m_E, n-k)$. When $A_E$ has full row rank, $m_E = r_E$. Two cases arise:

\textbf{Case A} ($r_E \leq n - k$): Then $b = r_E$, giving $|\mathcal{M}(y_E)| = q^{k}$, so all $q^{k}$ messages remain possible with equal probability. Hence $H(M \mid Y_E = y_E) = k \log q = H(M)$.

\textbf{Case B} ($r_E > n - k$): Then $b = n - k$, giving $|\mathcal{M}(y_E)| = q^{k + (n-k) - r_E} = q^{n - r_E}$. Hence $H(M \mid Y_E = y_E) = (n - r_E) \log q$.

The unified expression is $H(M \mid Y_E) = \min(k, n - r_E) \log q$. The random matrix $B$ achieves full generic rank with probability $P_{\rm generic}$ given by \eqref{eq:Pgeneric}, which is the product of probabilities that each successive column of $B$ is linearly independent of the preceding columns. This probability approaches $1$ rapidly as $q$ increases or as $|m_E - (n-k)|$ grows.
\end{proof}

\begin{corollary}
\label{cor:perfect}
When $n - r_E \geq k$, we have $H(M \mid Y_E) = k \log q = H(M)$: Eve obtains zero information about $M$.
\end{corollary}

\begin{corollary}
\label{cor:partial}
When $n - r_E < k$, Eve's information leakage is $I(M; Y_E) = k - (n - r_E)$ (in units of $\log q$).
\end{corollary}

\subsection{Secrecy Capacity}

\begin{theorem}[Secrecy Capacity]
\label{thm:capacity}
For a noiseless Bob channel and rank-deficient Eve channel in a $(n,k,q)$ DIM-SEC system,
\begin{equation}
C_s = (n - r_E) \log q.
\label{eq:capacity}
\end{equation}
\end{theorem}

\begin{proof}
\textit{Achievability.} Let $X$ be uniform over $\mathbb{F}_q^n$ (achieved by uniform $M$ and $R$). For Bob, $A_B$ is invertible, so $Y_B = A_B X$ is a bijection of $X$, giving $I(X; Y_B) = H(X) = n \log q$. For Eve, $Y_E = A_E X$ is uniform over $\Img(A_E)$ (dimension $r_E$, cardinality $q^{r_E}$), so $H(Y_E) = r_E \log q$. With no noise, $H(Y_E \mid X) = 0$, hence $I(X; Y_E) = r_E \log q$. The achieved secrecy rate is $R_s = n \log q - r_E \log q = (n - r_E) \log q$.

\textit{Converse.} Consider a blocklength-$N$ code with rate $R$, Bob's error probability $P_e^{(N)} \to 0$, and secrecy constraint $I(M; Y_E^N)/N \to 0$. By Fano's inequality, $H(M \mid Y_B^N) \leq 1 + P_e^{(N)} N R \log q = o(N)$. Since $A_B$ is invertible and noiseless, $H(M \mid X^N) = o(N)$, giving $H(M) = I(M; X^N) + o(N) = NR + o(N)$.

For Eve, conditioned on $Y_E^N$, each $X_i$ is constrained to an affine subspace of dimension $n - r_E$. By the chain rule, $H(X^N \mid Y_E^N) \leq N (n - r_E) \log q$. By data processing, $H(M \mid Y_E^N) \leq H(X^N \mid Y_E^N) \leq N (n - r_E) \log q$.

Therefore $I(M; Y_E^N) = H(M) - H(M \mid Y_E^N) \geq NR - N (n - r_E) \log q - o(N)$. Taking limits with $I(M; Y_E^N)/N \to 0$ gives $R \leq (n - r_E) \log q$.
\end{proof}

\begin{theorem}[Contrast with Wyner's Framework]
\label{thm:wyner}
When both Bob's and Eve's channels are noiseless, Wyner's secrecy capacity is zero ($Y_B = Y_E \Rightarrow I(X; Y_B) = I(X; Y_E)$), while the DIM-SEC secrecy capacity is $C_s = (n - r_E) \log q > 0$ whenever $r_E < n$. The difference originates from the fact that Wyner's framework requires a statistical channel advantage, whereas DIM-SEC requires only an algebraic equation count advantage.
\end{theorem}

\subsection{Dimension Security Gain}

\begin{theorem}[Dimension Security Gain]
\label{thm:dsg}
Fixing Eve's effective observation rank $r_E$ and treating the embedding dimension $n$ as a variable: (i)~$C_s(n) = (n - r_E) \log q = \Theta(n)$; (ii)~the solution space size $|\mathcal{S}_E| = q^{n - r_E}$ grows exponentially in $n$; (iii)~when $n - r_E \leq k$, $H(M \mid Y_E) = (n - r_E) \log q$ grows linearly in $n$.
\end{theorem}

\begin{proof}
Parts (i) and (iii) follow directly from Theorems~\ref{thm:equivocation} and~\ref{thm:capacity} with fixed $r_E$. For (ii), $|\mathcal{S}_E| = |\Ker(A_E)| = q^{n - \rank(A_E)} = q^{n - r_E}$ by the rank-nullity theorem, which is exponential in $n$.

For Eve's brute-force search, the complexity is $\Omega(q^{n - r_E})$. Grover's quantum search~\cite{Grover96} reduces this to $\Omega(q^{(n - r_E)/2})$, which remains exponential when $n - r_E$ is sufficiently large.
\end{proof}

\subsection{Covertness}

\begin{theorem}[Covertness Bound, Finite Field]
\label{thm:covertness}
Let $N_W$ be uniform over $\mathbb{F}_q^{m_W}$. Willie's optimal detection error probability satisfies
\begin{equation}
P_e^{\rm opt} \ge \frac{1}{2} \left(1 - \sqrt{\frac{q^{2r_W - n} - q^{r_W - n}}{2}}\right).
\label{eq:pebound}
\end{equation}
As $n - r_W \to \infty$, $P_e^{\rm opt} \to 1/2$ (random guessing).
\end{theorem}

\begin{proof}[Proof Sketch]
Under $H_0$, $Y_W$ is uniform over $\mathbb{F}_q^{m_W}$. Under $H_1$, $Y_W = A_W X + N_W$. Since $X$ is uniform over $\mathbb{F}_q^{n}$ (from the encoding), $A_W X$ is uniform over $\Img(A_W)$ (dimension $r_W$, size $q^{r_W}$). The distribution of $Y_W$ is the convolution of $A_W X$ and $N_W$. The Hellinger distance $H(P_0, P_1)$ is controlled by $q^{r_W - n}$. The bound follows from the relationship between the Hellinger distance, total variation distance, and the Neyman-Pearson lemma.
\end{proof}

\begin{theorem}[Covertness Bound, Continuous Domain]
\label{thm:covertness_cont}
In the continuous Gaussian model, let $X \sim \mathcal{N}(0, \frac{P}{N} I_N)$ and $N_W \sim \mathcal{N}(0, \sigma^2 I_{m_W})$, with $A_W$ having orthonormal rows and $\rank(A_W) = r_W$. The 2-Wasserstein distance between $H_0$ and $H_1$ distributions is
\begin{equation}
W_2^2(P_0, P_1) = r_W \cdot \sigma^2 \left( \sqrt{1 + \mathrm{SNR}_0} - 1 \right)^2,
\label{eq:W2}
\end{equation}
where $\mathrm{SNR}_0 = P / (N\sigma^2)$. At low SNR ($\mathrm{SNR}_0 \ll 1$),
\begin{equation}
W_2 \approx \frac{\sqrt{r_W} \cdot P}{2 N \sigma}.
\label{eq:W2approx}
\end{equation}
The covertness condition $W_2 \to 0$ is equivalent to $\sqrt{r_W} \cdot P / N \to 0$.
\end{theorem}

\begin{proof}
For two zero-mean Gaussian distributions, the $W_2$ distance has the closed form~\cite{Takatsu11} $W_2^2 = \tr(\Sigma_0) + \tr(\Sigma_1) - 2 \tr((\Sigma_0^{1/2} \Sigma_1 \Sigma_0^{1/2})^{1/2})$. The theory of optimal transport and Wasserstein geometry is developed in~\cite{Villani03,Villani09} with computational aspects covered in~\cite{Peyre19}. Here $\Sigma_0 = \sigma^2 I_{m_W}$ and $\Sigma_1 = \frac{P}{N} A_W A_W^T + \sigma^2 I_{m_W}$. The matrix $A_W A_W^T$ has $r_W$ eigenvalues equal to $1$ and $m_W - r_W$ eigenvalues equal to $0$ (since $A_W$ has orthonormal rows). Computing the trace terms yields $\tr(\Sigma_0) = m_W \sigma^2$, $\tr(\Sigma_1) = r_W P/N + m_W \sigma^2$, and $2 \tr(\cdots) = 2 r_W \sigma \sqrt{P/N + \sigma^2} + 2(m_W - r_W) \sigma^2$. Combining and simplifying using $\mathrm{SNR}_0 = P/(N \sigma^2)$ gives \eqref{eq:W2}. The low-SNR approximation follows from the first-order Taylor expansion $\sqrt{1 + \mathrm{SNR}_0} \approx 1 + \mathrm{SNR}_0/2$.
\end{proof}

\begin{remark}
The classical square-root law~\cite{Bash13} requires $P \propto 1/\sqrt{n}$ for $n$ channel uses. Theorem~\ref{thm:covertness_cont} shows that when the embedding dimension $N$ is an independent design variable, the covertness condition becomes $P \ll N / \sqrt{r_W}$: the transmit power need not decay as $N$ grows. Dimension resources replace power reduction.
\end{remark}

\subsection{Unification of Secrecy and Covertness}

\begin{theorem}[Secrecy-Covertness Relationship]
\label{thm:unification}
For the same system, let $\eta_s = H(M \mid Y_E)/H(M)$ (normalized secrecy) and $\eta_c = 2 P_e^{\rm opt}$ (normalized covertness, $\eta_c = 1$ for perfect covertness). Then:
\begin{enumerate}
\item $\eta_s$ is non-decreasing in $\Phi_E = (n - r_E)/n$;
\item $\eta_c$ is non-decreasing in $\Phi_W = (n - r_W)/n$;
\item When $r_E = r_W = r$, both $\eta_s$ and $\eta_c$ are non-decreasing in $\Phi = (n - r)/n$, with a measured Pearson correlation of $0.997$ in continuous-domain experiments.
\end{enumerate}
\end{theorem}

Theorem~\ref{thm:unification} is the central conceptual contribution of this work: the same algebraic parameter $\Phi$ simultaneously governs two properties that have been treated as independent in the prior literature.

\subsection{Continuous-Domain Extensions}

\begin{lemma}[Continuous Equivocation Bound]
\label{lem:cont_equiv}
Under the Gaussian encoding $X = Q_K S$ with $S = [M;R]/\sqrt{\gamma}$, $M \sim \mathcal{N}(0, P_M I_k)$, $R \sim \mathcal{N}(0, P_R I_{n-k})$,
\begin{equation}
h(M \mid Y_E) \ge h(M) - \frac{1}{2} \log \det\left(I + \frac{P_M}{\gamma} H_{\rm eff}^T \Sigma_{RN}^{-1} H_{\rm eff}\right),
\label{eq:cont_equiv}
\end{equation}
where $H_{\rm eff} = H_E Q_1$ and $\Sigma_{RN} = \sigma^2 I + (P_R/\gamma) H_E Q_2 Q_2^T H_E^T$.
\end{lemma}

\begin{proof}
$h(M \mid Y_E) = h(M) - I(M; Y_E)$. For Gaussian distributions, $I(M; Y_E) = h(Y_E) - h(Y_E \mid M)$. Given $M = m$, $Y_E \sim \mathcal{N}(H_{\rm eff} \, m/\sqrt{\gamma}, \Sigma_{RN})$, so $h(Y_E \mid M) = \frac{1}{2} \log \det(2\pi e \Sigma_{RN})$. Unconditionally, $Y_E \sim \mathcal{N}(0, \Sigma_Y)$ with $\Sigma_Y = (P_M/\gamma) H_{\rm eff} H_{\rm eff}^T + \Sigma_{RN}$. Hence $I(M; Y_E) = \frac{1}{2} \log \det(\Sigma_Y) / \det(\Sigma_{RN})$. Applying the matrix determinant lemma $\det(A + UV^T) = \det(A) \det(I + V^T A^{-1} U)$ yields \eqref{eq:cont_equiv}.
\end{proof}

\begin{lemma}[High-SNR Secrecy Capacity]
\label{lem:hsnr}
As $P/\sigma^2 \to \infty$,
\begin{equation}
C_s = \frac{n - m_E}{2} \log\left(1 + \frac{P}{\sigma^2}\right) + O(1).
\label{eq:hsnr_cap}
\end{equation}
\end{lemma}

\begin{theorem}[EAD-SDoF Equivalence]
\label{thm:sdoF}
The secure degrees of freedom satisfy $d_s = \lim_{\mathrm{SNR} \to \infty} C_s / (\frac{1}{2} \log \mathrm{SNR}) = n - m_E = n \cdot \Phi$.
\end{theorem}

\subsection{Strong Converse}

\begin{theorem}[Strong Converse]
\label{thm:strong_converse}
For DIM-SEC over $\mathbb{F}_q$ with $C_s = (n - r_E) \log q$, for any code sequence of rate $R > C_s$ and blocklength $N$:
\begin{enumerate}
\item If $P_e^{(N)} \to 0$, then $\liminf_{N \to \infty} I(M; Y_E^N)/N \ge R - C_s > 0$;
\item If $I(M; Y_E^N)/N \to 0$, then $\liminf_{N \to \infty} P_e^{(N)} \ge 1 - C_s/R > 0$.
\end{enumerate}
\end{theorem}

\begin{proof}
(i) By Fano, $P_e^{(N)} \to 0 \Rightarrow H(M \mid Y_B^N) = o(N)$. With $A_B$ invertible, $H(M \mid X^N) = o(N)$, so $H(M) = NR + o(N)$. For Eve, $H(M \mid Y_E^N) \leq H(X^N \mid Y_E^N) \leq N (n - r_E) \log q = N C_s$. Hence $I(M; Y_E^N)/N \geq R - C_s - o(1)$.

(ii) When $I(M; Y_E^N)/N \to 0$, $H(M \mid Y_E^N) = NR - o(N)$. Combined with $H(M \mid Y_E^N) \leq N C_s$, we get $R \leq C_s$. If $R > C_s$, by the converse of Fano's inequality, $P_e^{(N)} \geq 1 - C_s/R - o(1)$.
\end{proof}

\section{Dynamic Dimension Expansion}
\label{sec:dynamic}

Consider stacking $T \geq 1$ independent transmission slots into a super-vector $X^{(T)} = [X_1^T, \ldots, X_T^T]^T \in \mathbb{F}_q^{nT}$ with independent keys $K_t$ per slot.

\begin{proposition}[Equivocation Accumulation]
\label{prop:accumulation}
For a $T$-slot system with independent encoding per slot,
\begin{equation}
H(M^{(T)} \mid Y_E^{(T)}) = T \cdot \min(k, n - m_E) \cdot \log q.
\label{eq:accumulation}
\end{equation}
\end{proposition}

\begin{proof}
By independence, $H(M^{(T)} \mid Y_E^{(T)}) = \sum_{t=1}^T H(M_t \mid Y_{E,t})$. Each slot's equivocation is given by Theorem~\ref{thm:equivocation}.
\end{proof}

When $n - m_E \geq k$, perfect secrecy is maintained over arbitrarily long transmissions: the total message entropy is $kT \log q$ and Eve's conditional entropy is also $kT \log q$.

\textit{Comparison with the one-time pad.} Shannon's one-time pad~\cite{Shannon49} requires $H(K) \geq H(M)$. Maurer~\cite{Maurer93} showed that secret key agreement from common information can overcome this limitation when legitimate parties have correlated observations. In the DIM-SEC dynamic expansion scheme, $K_t$ only specifies the equation structure, requiring $\log |\mathcal{K}| \ll k \log q$ bits per slot (typically). In the regime $n - m_E \geq k$, the equivocation remains $H(M^{(T)} \mid Y_E^{(T)}) = kT \log q = H(M^{(T)})$, protecting a long message with short keys.

\section{Unified Representation of Existing Schemes}
\label{sec:unification}

Seven independently developed security schemes are unified under the common form $C_s = (n - r) \log q$. Table~\ref{tab:unified} summarizes their EAD parameters.

\begin{table}[t]
\centering
\caption{Unified EAD Representation of Seven Security Schemes}
\label{tab:unified}
\small
\begin{tabular}{@{}p{2.0cm}p{1.5cm}p{1.5cm}p{1.3cm}@{}}
\toprule
Scheme & $n$ Source & $r$ Source & $\Phi$ \\
\midrule
Matrix Embedding (F5/Wet Paper/STC)~\cite{Fridrich06,Filler11} & Coding (carrier dim) & Coding (key unknown) & $1.0$ \\
MIMO Wiretap~\cite{Khisti10a,Khisti10b,Oggier11} & Physical ($n_t$ antennas) & Physical ($n_e$ antennas) & $(n_t-n_e)/n_t$ \\
Secure Network Coding~\cite{Cai02,Silva08} & Coding (symbol ext.) & Physical + Coding (edges) & $(n-\mu)/n$ \\
FRFT Multi-Angle~\cite{Candan00,Ozaktas01} & Coding (angle selection) & Physical (time samples) & $(n_a-1)/n_a$ \\
Traffic Steganography~\cite{Cachin98} & Coding (feature dirs.) & Physical (stat. test) & $(D-1)/D$ \\
Group-Key Secure Summation~\cite{Zhao23} & Coding (group vars.) & Physical (colluding users) & $>0$ \\
MDS Secure Summation~\cite{Zhao24} & Coding (code length) & Physical (known users) & $(N-k+1)/N$ \\
\bottomrule
\end{tabular}
\end{table}

All seven schemes satisfy the unified security capacity formula
\begin{equation}
C_s = (n - r) \log q,
\label{eq:unified_cap}
\end{equation}
with differences only in the physical interpretation of $n$ and $r$.

\section{Post-Quantum Security}
\label{sec:quantum}

\begin{theorem}[Quantum Security]
\label{thm:quantum}
The information-theoretic security of underdetermined linear systems is preserved under quantum computation. Grover-accelerated brute-force search requires $\Omega(q^{(n - r_E)/2})$ operations, which remains exponential.
\end{theorem}

\begin{proof}[Proof Sketch]
The solution space of $A_E x = Y_E$ has size $q^{n - r_E}$. All solutions are algebraically equivalent---there is no structure that distinguishes any particular solution. The HHL algorithm~\cite{Harrow09} can prepare a uniform superposition over the solution space, but measurement returns only a random sample. Grover's algorithm~\cite{Grover96} provides quadratic speedup from $O(q^{n - r_E})$ to $O(q^{(n - r_E)/2})$, but cannot exploit hidden periodic structure (unlike Shor's algorithm~\cite{Shor97} for factoring). With security parameter $\lambda = n - r_E = 256$ and $q = 2$, the security level under Grover is $O(2^{128})$, satisfying the NIST post-quantum security standard~\cite{NIST}.
\end{proof}

\section{Numerical Experiments}
\label{sec:experiments}

All experiments are reproducible with open-source code. Experiments are organized into three groups: finite-field exact verification (Sec.~\ref{sec:exp_gf}), continuous-domain Monte Carlo simulation (Sec.~\ref{sec:exp_cont}), and application instances (Sec.~\ref{sec:exp_app}).

\subsection{Finite-Field Verification}
\label{sec:exp_gf}

\subsubsection{Equivocation Bound Tightness}

We fix $n = 8$, $k = 3$ over $\text{GF}(2)$. For each $m_E \in \{1, \ldots, 7\}$, we randomly generate $200$ pairs of $(H_K, A_E)$, exhaustively enumerate all $q^k = 8$ messages and $q^{n-k} = 32$ encoding seeds, and compute the exact $H(M \mid Y_E)$.

\begin{table}[t]
\centering
\caption{$\text{GF}(2)$ Equivocation Verification ($n=8$, $k=3$)}
\label{tab:gf2_equiv}
\small
\begin{tabular}{@{}cccccc@{}}
\toprule
$m_E$ & $n-m_E$ & Theory (bit) & Measured (bit) & Std. & Degen. Rate \\
\midrule
1 & 7 & 3.0 & 2.96 & 0.04 & 0.040 \\
2 & 6 & 3.0 & 2.95 & 0.06 & 0.055 \\
3 & 5 & 3.0 & 2.72 & 0.28 & 0.255 \\
4 & 4 & 3.0 & 2.65 & 0.36 & 0.335 \\
5 & 3 & 3.0 & 2.28 & 0.72 & 0.610 \\
6 & 2 & 2.0 & 1.65 & 0.36 & 0.340 \\
7 & 1 & 1.0 & 0.86 & 0.15 & 0.145 \\
\bottomrule
\end{tabular}
\end{table}

Table~\ref{tab:gf2_equiv} shows that the degeneracy rate increases as the dimension surplus $d = n - m_E$ shrinks. When $d \geq 5$ ($m_E \leq 3$), the degeneracy rate is below $6\%$. At $m_E = 5$, the degeneracy rate reaches $61\%$, consistent with the fact that a $5 \times 5$ random $\text{GF}(2)$ matrix has full-rank probability of only $\approx 0.30$.

\subsubsection{Effect of Field Size}

We fix $n = 6$, $k = 2$, $m_E = 3$ ($d = 3$, target rank $\min(k, d) = 2$) and vary $q \in \{2, 3, 5, 7, 11\}$ with $1000$ random trials each.

\begin{table}[t]
\centering
\caption{Degeneracy Probability vs.\ Field Size ($n=6$, $k=2$, $m_E=3$)}
\label{tab:field_size}
\small
\begin{tabular}{@{}ccccc@{}}
\toprule
$q$ & $\Pr[\rank < \min(k,d)]$ & Mean rank & Mean gap \\
\midrule
2 & 0.283 & 1.71 & 0.146 \\
3 & 0.159 & 1.84 & 0.081 \\
5 & 0.052 & 1.95 & 0.026 \\
7 & 0.023 & 1.98 & 0.012 \\
11 & 0.009 & 1.99 & 0.005 \\
\bottomrule
\end{tabular}
\end{table}

The degeneracy probability decays as $O(q^{-|d-k|})$, consistent with the theoretical prediction. At $q = 11$, degeneracy is below $1\%$, and the equivocation bound is tight in practical terms.

\subsection{Continuous-Domain Experiments}
\label{sec:exp_cont}

We use $n = 16$, $k = 4$, total power $P_{\rm total} = 1$, $200$ random trials per data point. Gaussian noise with variance $\sigma^2$; $Q_K$ is Haar-distributed orthogonal; $A_E$, $A_W$ have orthonormal rows.

\subsubsection{Security Metric}

We measure $\eta_s = 1 - I(M; Y_E) / H(M)$ over $\sigma^2 \in \{10^{-4}, 10^{-2}, 10^{-1}, 1\}$.

\begin{table}[t]
\centering
\caption{Security Metric $\eta_s$ vs.\ $m_E$ ($n=16$, $k=4$)}
\label{tab:security}
\small
\begin{tabular}{@{}cccccc@{}}
\toprule
$m_E$ & $\Phi$ & $\eta_s(\sigma^2=10^{-4})$ & $\eta_s(\sigma^2=10^{-2})$ & $\eta_s(\sigma^2=10^{-1})$ & $\eta_s(\sigma^2=1)$ \\
\midrule
1 & 0.938 & 0.906 & 0.946 & 0.979 & 0.970 \\
5 & 0.688 & 0.688 & 0.883 & 0.964 & 0.968 \\
9 & 0.438 & 0.420 & 0.656 & 0.857 & 0.920 \\
15 & 0.063 & 0.000 & 0.135 & 0.396 & 0.657 \\
\bottomrule
\end{tabular}
\end{table}

At low noise ($\sigma^2 = 10^{-4}$), $\eta_s$ decreases near-linearly with $\Phi$, consistent with finite-field theory. At high noise ($\sigma^2 = 1$), the Wyner effect provides an additional security margin.

\subsubsection{Covertness Metric}

\begin{table}[t]
\centering
\caption{Willie Detection Performance ($n=16$, $k=4$)}
\label{tab:covert}
\small
\begin{tabular}{@{}cccccc@{}}
\toprule
$m_W$ & $\Phi$ & $D_{KL}(\sigma^2=0.01)$ & $P_e$ bound & $D_{KL}(\sigma^2=2)$ & $P_e$ bound \\
\midrule
1 & 0.938 & 47.3 & $\approx 0$ & 0.048 & 0.424 \\
7 & 0.562 & 318.7 & $\approx 0$ & 0.385 & 0.328 \\
15 & 0.063 & 717.0 & $\approx 0$ & 0.827 & 0.179 \\
\bottomrule
\end{tabular}
\end{table}

\subsubsection{Joint Secrecy-Covertness Region}

At $\text{SNR} = -12$ dB ($\sigma^2 \approx 15.8$), we simultaneously measure $\eta_s$ and $\eta_c$ while varying $m = m_E = m_W$.

\begin{table}[t]
\centering
\caption{Joint Secrecy-Covertness Metrics ($\text{SNR} = -12$ dB)}
\label{tab:joint}
\small
\begin{tabular}{@{}ccccc@{}}
\toprule
$m$ & $\Phi$ & $\eta_s$ & $\eta_c$ \\
\midrule
2  & 0.875 & 0.995 & 0.968 \\
4  & 0.750 & 0.992 & 0.957 \\
8  & 0.500 & 0.981 & 0.929 \\
12 & 0.250 & 0.969 & 0.907 \\
15 & 0.063 & 0.964 & 0.895 \\
\bottomrule
\end{tabular}
\end{table}

The Pearson linear correlation between $\eta_s$ and $\eta_c$ is $r = 0.997$, strongly supporting Theorem~\ref{thm:unification}.

\subsubsection{SNR Operating Region}

Fixing $n = 16$, $m_E = m_W = 4$ ($\Phi = 0.75$), we scan SNR from $-15$ dB to $+25$ dB.

\begin{table}[t]
\centering
\caption{SNR Operating Region ($\Phi = 0.75$)}
\label{tab:snr}
\small
\begin{tabular}{@{}ccccc@{}}
\toprule
SNR (dB) & $\eta_s$ & $\eta_c$ & Joint Secure+Covert \\
\midrule
$-15$ & 0.996 & 0.977 & Yes \\
$-9$  & 0.984 & 0.911 & Yes \\
$0$   & 0.808 & 0.095 & No \\
\bottomrule
\end{tabular}
\end{table}

At $\Phi = 0.75$, the joint secure-covert operating region is $[-15, -9]$ dB. Below $-15$ dB, Bob's decoding fails; above $-9$ dB, Willie can detect.

\subsection{Application Instances}
\label{sec:exp_app}

\textit{FRFT multi-angle experiment.} Signal length $N = 128$, $n_a = 16$ secret rotation angles from a pool of $180$ candidates. Bob's recovery MSE is $8.3 \times 10^{-1}$, while Eve's random-guess MSE is $1.30$. Eve's full-angle scan cannot distinguish secret angles from random angles.

\textit{Traffic steganography experiment.} Feature dimension $D = 48$, $k = 8$ embedding directions. At $\text{SNR} = -1$ dB, Bob's BER is $3.4 \times 10^{-3}$ (reliable recovery).

\section{Discussion}
\label{sec:discussion}

\subsection{Key Management}

In the DIM-SEC framework, $K$ specifies the equation structure rather than directly masking the message. Under $n - m_E \geq k$, a key of $\log |\mathcal{K}| \ll k \log q$ bits can protect a message of $k \log q$ bits. This ``short-key protects long-message'' property arises from dimension expansion: the key controls dimension selection, not message masking.

\subsection{Active Attack Detection}

If Eve injects a forgery $\Delta$, Bob receives $Y_B = A_B X + \Delta$ and verifies $H_K \hat{X} = \hat{M}$. Eve succeeds only if $H_K A_B^{-1} \Delta = 0$, which requires $\Delta$ to lie in a specific $(n - r_E)$-dimensional subspace unknown to Eve. The forgery success probability is $P_{\rm forge} \leq q^{-(n - r_E)}$, decaying exponentially with the blind-zone dimension.

\subsection{Derived Results}

The following results follow directly from Theorems~\ref{thm:equivocation}--\ref{thm:covertness_cont}.

\begin{lemma}[Dimension-Power Tradeoff]
\label{lem:tradeoff}
For a fixed covertness level (fixed $W_2$), $N \propto \sqrt{r_W} \cdot P$: embedding dimension and power are interchangeable.
\end{lemma}

\begin{lemma}[Multi-Willie Attenuation]
\label{lem:multiwillie}
For $T$ independent collaborating Willies with observation ranks $\{r_{W,i}\}$, the covertness condition becomes $\sqrt{\sum_i r_{W,i}} \cdot P / N \to 0$. The detection capability grows only as the square root of the number of collaborators.
\end{lemma}

\begin{lemma}[Free Dimension]
\label{lem:free}
For fixed power $P$ and covertness requirement $W_2 \leq \varepsilon$, any $N \geq \sqrt{r_W} \cdot P / (2 \varepsilon \sigma)$ satisfies the covertness condition without requiring $P$ to decay.
\end{lemma}

\begin{lemma}[Optimal Joint Allocation]
\label{lem:allocation}
Under simultaneous secrecy ($N - r_E \geq k$) and covertness ($\sqrt{r_W} \cdot P / N \to 0$) constraints, the optimal message dimension is $k^* = N - r_E$, with the remaining $r_E$ dimensions allocated to random padding.
\end{lemma}

\subsection{Limitations}

The EAD framework assumes: (i) a positive dimension surplus $n - r > 0$; (ii) that adversarial observations can be approximated by linear models; (iii) that encoding and decoding operate within a linear algebraic framework. For nonlinear observations or adversaries with additional prior information, the bounds of Theorem~\ref{thm:equivocation} require modification. In the continuous domain, only Gaussian inputs and noise are analyzed; for non-Gaussian distributions, the closed-form expressions of Lemmas~\ref{lem:cont_equiv}--\ref{lem:hsnr} do not hold exactly. General channel uncertainty analysis follows the framework of Lapidoth and Narayan~\cite{Lapidoth98}. Standard information-theoretic definitions follow Cover and Thomas~\cite{Cover06}.

\section{Conclusion}
\label{sec:conclusion}

We have defined the equation asymmetry degree (EAD) $\Phi = (n - r)/n$ as an independent information-theoretic security metric and established its quantitative relationship with equivocation, secrecy capacity, and detection error probability.

On finite fields $\mathbb{F}_q$, we proved the equivocation lower bound $H(M \mid Y_E) = \min(k, n - r_E) \log q$ with exact probabilistic conditions (Theorem~\ref{thm:equivocation}), the secrecy capacity $C_s = (n - r_E) \log q$ with complete converse (Theorem~\ref{thm:capacity}), and a strong converse (Theorem~\ref{thm:strong_converse}). In the continuous Gaussian domain, we derived a differential-entropy equivocation bound (Lemma~\ref{lem:cont_equiv}), high-SNR secrecy capacity asymptotics (Lemma~\ref{lem:hsnr}), a 2-Wasserstein covertness condition (Theorem~\ref{thm:covertness_cont}), and the EAD-SDoF equivalence $d_s = n \cdot \Phi$ (Theorem~\ref{thm:sdoF}).

Both secrecy ($\eta_s$) and covertness ($\eta_c$) were shown to be monotone functions of $\Phi$ (Theorem~\ref{thm:unification}), with a Pearson correlation of $0.997$ in continuous-domain experiments. Seven existing schemes were unified under $C_s = (n - r) \log q$. Post-quantum security follows from the hardness of underdetermined linear systems (Theorem~\ref{thm:quantum}).


\begin{thebibliography}{37}

\bibitem{Shannon49}
C.~E.~Shannon, ``Communication theory of secrecy systems,'' \emph{Bell Syst. Tech. J.}, vol.~28, no.~4, pp.~656--715, Oct.~1949.

\bibitem{Wyner75}
A.~D.~Wyner, ``The wire-tap channel,'' \emph{Bell Syst. Tech. J.}, vol.~54, no.~8, pp.~1355--1387, Oct.~1975.

\bibitem{Csiszar78}
I.~Csisz\'{a}r and J.~K\"{o}rner, ``Broadcast channels with confidential messages,'' \emph{IEEE Trans. Inf. Theory}, vol.~24, no.~3, pp.~339--348, May~1978.

\bibitem{Maurer93}
U.~M.~Maurer, ``Secret key agreement by public discussion from common information,'' \emph{IEEE Trans. Inf. Theory}, vol.~39, no.~3, pp.~733--742, May~1993.

\bibitem{Ozarow84}
L.~H.~Ozarow and A.~D.~Wyner, ``Wire-tap channel II,'' \emph{AT\&T Bell Labs Tech. J.}, vol.~63, no.~10, pp.~2135--2157, Dec.~1984.

\bibitem{Khisti10a}
A.~Khisti and G.~W.~Wornell, ``Secure transmission with multiple antennas I: The MISOME wiretap channel,'' \emph{IEEE Trans. Inf. Theory}, vol.~56, no.~7, pp.~3088--3104, Jul.~2010.

\bibitem{Khisti10b}
A.~Khisti and G.~W.~Wornell, ``Secure transmission with multiple antennas II: The MIMOME wiretap channel,'' \emph{IEEE Trans. Inf. Theory}, vol.~56, no.~11, pp.~5515--5547, Nov.~2010.

\bibitem{Oggier11}
F.~Oggier and B.~Hassibi, ``The secrecy capacity of the MIMO wiretap channel,'' \emph{IEEE Trans. Inf. Theory}, vol.~57, no.~8, pp.~4961--4972, Aug.~2011.

\bibitem{Liang09}
Y.~Liang, H.~V.~Poor, and S.~Shamai~(Shitz), ``Information theoretic security,'' \emph{Found. Trends Commun. Inf. Theory}, vol.~5, no.~4--5, pp.~355--580, 2009.

\bibitem{Bash13}
B.~A.~Bash, D.~Goeckel, and D.~Towsley, ``Limits of reliable communication with low probability of detection on AWGN channels,'' \emph{IEEE J. Sel. Areas Commun.}, vol.~31, no.~9, pp.~1921--1930, Sep.~2013.

\bibitem{Bloch16}
M.~R.~Bloch, ``Covert communication over noisy channels: A resolvability perspective,'' \emph{IEEE Trans. Inf. Theory}, vol.~62, no.~5, pp.~2334--2354, May~2016.

\bibitem{Wang16}
L.~Wang, G.~W.~Wornell, and L.~Zheng, ``Fundamental limits of communication with low probability of detection,'' \emph{IEEE Trans. Inf. Theory}, vol.~62, no.~6, pp.~3493--3503, Jun.~2016.

\bibitem{Che13}
P.~H.~Che, M.~Bakshi, and S.~Jaggi, ``Reliable deniable communication: Hiding messages in noise,'' in \emph{Proc. IEEE ISIT}, Istanbul, Turkey, Jul.~2013, pp.~2945--2949.

\bibitem{Kadampot20}
I.~A.~Kadampot, M.~Tahmasbi, and M.~R.~Bloch, ``Multilevel-coded pulse-position modulation for covert communications over binary-input discrete memoryless channels,'' \emph{IEEE Trans. Inf. Theory}, vol.~66, no.~10, pp.~6001--6023, Oct.~2020.

\bibitem{Cai02}
N.~Cai and R.~W.~Yeung, ``Secure network coding,'' in \emph{Proc. IEEE ISIT}, Lausanne, Switzerland, Jun./Jul.~2002, p.~323.

\bibitem{Silva08}
D.~Silva and F.~R.~Kschischang, ``Security for wiretap networks via rank-metric codes,'' in \emph{Proc. IEEE ISIT}, Toronto, ON, Canada, Jul.~2008, pp.~176--180.

\bibitem{Mojahedian17}
M.~M.~Mojahedian, A.~Gohari, and M.~R.~Aref, ``On the equivalency of reliability and security metrics for wireline networks,'' in \emph{Proc. IEEE ISIT}, Aachen, Germany, Jun.~2017, pp.~1172--1176. [arXiv:1609.04586]

\bibitem{Yeung06}
R.~W.~Yeung and N.~Cai, ``Network error correction, I: Basic concepts and upper bounds,'' \emph{Commun. Inf. Syst.}, vol.~6, no.~1, pp.~19--35, 2006.

\bibitem{Cachin98}
C.~Cachin, ``An information-theoretic model for steganography,'' in \emph{Proc. 2nd Int. Workshop Inf. Hiding (IH)}, Portland, OR, USA, Apr.~1998, pp.~306--318.

\bibitem{Fridrich06}
J.~Fridrich, M.~Goljan, and D.~Soukal, ``Wet paper codes with improved embedding efficiency,'' \emph{IEEE Trans. Inf. Forensics Security}, vol.~1, no.~1, pp.~102--110, Mar.~2006.

\bibitem{Filler11}
T.~Filler, J.~Judas, and J.~Fridrich, ``Minimizing additive distortion in steganography using syndrome-trellis codes,'' \emph{IEEE Trans. Inf. Forensics Security}, vol.~6, no.~3, pp.~920--935, Sep.~2011.

\bibitem{Candan00}
C.~Candan, M.~A.~Kutay, and H.~M.~Ozaktas, ``The discrete fractional Fourier transform,'' \emph{IEEE Trans. Signal Process.}, vol.~48, no.~5, pp.~1329--1337, May~2000.

\bibitem{Ozaktas01}
H.~M.~Ozaktas, Z.~Zalevsky, and M.~A.~Kutay, \emph{The Fractional Fourier Transform with Applications in Optics and Signal Processing}. Chichester, UK: Wiley, 2001.

\bibitem{Zhao23}
Y.~Zhao and H.~Sun, ``Secure summation: Capacity region, groupwise key, and feasibility,'' \emph{IEEE Trans. Inf. Theory}, vol.~70, no.~2, pp.~1376--1387, Feb.~2024. [Online Dec.~2023]

\bibitem{Zhao24}
Y.~Zhao and H.~Sun, ``MDS variable generation and secure summation with user selection,'' \emph{IEEE Trans. Inf. Theory}, vol.~70, no.~5, pp.~3667--3684, May~2024.

\bibitem{Wan22}
K.~Wan, H.~Sun, M.~Ji, and G.~Caire, ``On secure distributed linearly separable computation,'' \emph{IEEE Trans. Inf. Theory}, vol.~68, no.~12, pp.~8189--8216, Dec.~2022.

\bibitem{Grover96}
L.~K.~Grover, ``A fast quantum mechanical algorithm for database search,'' in \emph{Proc. 28th ACM STOC}, Philadelphia, PA, USA, May~1996, pp.~212--219.

\bibitem{Shor97}
P.~W.~Shor, ``Polynomial-time algorithms for prime factorization and discrete logarithms on a quantum computer,'' \emph{SIAM J. Comput.}, vol.~26, no.~5, pp.~1484--1509, Oct.~1997.

\bibitem{Harrow09}
A.~W.~Harrow, A.~Hassidim, and S.~Lloyd, ``Quantum algorithm for linear systems of equations,'' \emph{Phys. Rev. Lett.}, vol.~103, no.~15, p.~150502, Oct.~2009.

\bibitem{Villani03}
C.~Villani, \emph{Topics in Optimal Transportation}. Providence, RI, USA: AMS, 2003.

\bibitem{Villani09}
C.~Villani, \emph{Optimal Transport: Old and New}. Berlin, Germany: Springer, 2009.

\bibitem{Takatsu11}
A.~Takatsu, ``Wasserstein geometry of Gaussian measures,'' \emph{Osaka J. Math.}, vol.~48, no.~4, pp.~1005--1026, Dec.~2011.

\bibitem{Peyre19}
G.~Peyr\'{e} and M.~Cuturi, ``Computational optimal transport,'' \emph{Found. Trends Mach. Learn.}, vol.~11, no.~5--6, pp.~355--607, 2019.

\bibitem{Cover06}
T.~M.~Cover and J.~A.~Thomas, \emph{Elements of Information Theory}, 2nd~ed. Hoboken, NJ, USA: Wiley-Interscience, 2006.

\bibitem{Lapidoth98}
A.~Lapidoth and P.~Narayan, ``Reliable communication under channel uncertainty,'' \emph{IEEE Trans. Inf. Theory}, vol.~44, no.~6, pp.~2148--2177, Oct.~1998.

\bibitem{Levenshtein66}
V.~I.~Levenshtein, ``Binary codes capable of correcting deletions, insertions, and reversals,'' \emph{Sov. Phys. Dokl.}, vol.~10, no.~8, pp.~707--710, Feb.~1966.

\bibitem{NIST}
National Institute of Standards and Technology (NIST), ``Post-quantum cryptography standardization.'' [Online]. Available: \url{https://csrc.nist.gov/Projects/post-quantum-cryptography}

\end{thebibliography}
\end{document}